\documentclass[a4paper]{prepdmuc}

\usepackage{amssymb}
\usepackage{amsmath}
\usepackage{amsfonts}
\usepackage{amssymb}
\usepackage{amsthm}
\usepackage{graphicx}
\usepackage{slashbox}

\newtheorem{theorem}{Theorem}
\newtheorem{lemma}[theorem]{Lemma}
\newtheorem{proposition}[theorem]{Proposition}
\newtheorem{corollary}[theorem]{Corollary}

\theoremstyle{definition}

\newtheorem{definition}[theorem]{Definition}
\newtheorem{example}[theorem]{Example}

\newcommand{\KS}{{\rm KS}}

\title{Second-order stochastic comparisons of order statistics}

\autor[T. Lando]{Tommaso Lando}{SAEMQ, University of Bergamo (Italy) and V\v{S}B-TU Ostrava (Czech Republic)}{Tommaso Lando was supported by the Italian funds ex MURST 60\% 2019, by the Czech Science Foundation (GACR) under project 20-16764S and moreover by SP2020/11, an SGS research project of V\v{S}B-TU Ostrava. The support is greatly acknowledged.}
\email{tommaso.lando@unibg.it}

\autor[I. Arab]{Idir Arab}{CMUC, Department of Mathematics, University of Coimbra, Portugal}{This work was partially supported by the Centre for Mathematics of the University of Coimbra - UIDB/00324/2020, funded by the Portuguese Government through FCT/MCTES.}
\email{idir.bhh@gmail.com}

\autor[P.E. Oliveira]{Paulo Eduardo Oliveira} 
{CMUC, Department of Mathematics, University of Coimbra, Portugal}{}
\email{paulo@mat.uc.pt}

\date{May 25, 2020}
\infonum{20}{19}

\begin{document}

\begin{abstract}
We study the problem of comparing ageing patterns of the lifetime of $k$-out-of-$n$ systems. Mathematically, this reduces to being able to decide about a stochastic ordering relationship between different order statistics. We discuss such relationships with respect to second-order stochastic dominance, obtaining characterizations through the verification of relative convexity with respect to a suitably chosen reference distribution function. We introduce a hierarchy of such reference functions leading to classes, each expressing different and increasing knowledge precision about the distribution of the component lifetimes. Such classes are wide enough to include popular families of distributions, such as, for example, the increasing failure rate distributions. We derive sufficient dominance conditions depending on the identification of the class which includes the component lifetimes. Concerning the conditions, as expected, relying on a larger class of distributions, meaning that we have less precise information about the components' behaviour, leads to the need of stronger assumptions. We discuss the applicability of this method and characterize a test for the relative convexity, as this notion plays a central role in the proposed approach.
\end{abstract}

\maketitle

\section{Introduction}
In this paper we deal with the problem of comparing the lifetimes of $k$-out-of-$n$ systems with respect to ageing properties, by relying on the theory of \textit{stochastic orders}, Shaked and Shanthikumar~\cite{shaked2007}. We recall that the lifetime of a $k$-out-of-$n$ system is represented by the waiting time until fewer than $k$ components remain functioning in a system of $n$ components. Within a probabilistic framework, if we assume that the lifetime of each component is distributed according to a common, or \textit{parent}, cumulative distribution function (CDF), say $F$, then the lifetime of the system is represented by the \textit{order statistic} $X_{k:n}$, corresponding to a random sample of size $n$ from $F$.
For this reason, stochastic comparisons of order statistics represent a major issue in reliability theory. Engineering is typically concerned with choosing the system which may provide the best performance, according to some characteristics. Similarly to most decision problems in other research fields (e.g., economics, finance, etc.), the ``best'' performance of a $k$-out-of-$n$ system is generally understood as i) larger \textit{magnitude}, to be understood as the tendency of one random variable (RV) to take larger values, and ii) smaller risk or dispersion, since lifetime predictability is always preferable in such a context.
In reliability, the main stochastic orders used for comparisons of order statistics are \textit{likelihood ratio order, hazard rate order, first-order stochastic dominance} (FSD), with regard to \textit{magnitude} problems, Lillo~\cite{lillo2001}, Shaked and Shanthikumar~\cite{shaked2007} or Kochar~\cite{kochar2012}, and \textit{convex transform order, star order, Lorenz order, dispersive order}, Arnold and Villase\~{n}or~\cite{arnold1991}, Arnold and Nagaraja~\cite{arnold1991exp}, Wilfling~\cite{wilfling1996c}, Kochar~\cite{kochar2006}, Kochar and Xu~\cite{kochar2014} or Wu etal.~\cite{wu2020}, to deal with dispersion characterizations.
Recently, Lando and Bertoli-Barsotti~\cite{lando2019} considered the problem of ranking order statistics via \textit{second-order stochastic dominance} (SSD), that is the most widely used stochastic order in areas such as economics, finance, decision science and management. As well known, SSD, also referred to as generalized Lorenz dominance, is a scale-dependent version of the Lorenz order, which enables comparisons of RVs in terms of both magnitude and dispersion, therefore combining aspects i) and ii) into a single preorder. In this paper, we focus on the derivation of SSD for order statistics, dealing with both the one-sample (same parent distribution) and the two-sample problems (different parent distributions). For technical reasons, most results in the aforementioned literature about stochastic comparisons of order statistics are obtained by imposing restrictive constraints on the parent's distribution shape, or by focusing on particular parametric families of parent distributions.
However, restrictive shape assumptions are rather inconsistent with modern nonparametric statistical approach, in which the parent distribution is supposed to be unknown and has to be estimated from the data, with no prior constraint on its mathematical form. To address this issue, we propose a general method to derive SSD conditions for order statistics, according to different assumptions on the parent distribution.

Let $F$ and $H$ be continuous CDFs.
Renaming the definition of the convex transform order of van Zwet~\cite{zwet1964}, we shall say that $F$ is $H^{-1}$--\textit{convex} iff $H^{-1} \circ F$ is convex.
For the one-sample problem, we propose a method to compare, with respect to SSD, the order statistics $X_{i:n}$ and $X_{j:m}$ with common parent distribution $F$, by assuming that $F$ is $H^{-1}$--convex w.r.t. some suitable function $H^{-1}$. By focusing on four convenient choices of $H^{-1}$, we determine four partially nested classes of parent distributions that are relevant in terms of reliability properties, among which we can mention the well known \textit{increasing failure rate} (IFR) class Barlow et al.~\cite{barlow1963}. Correspondingly, we obtain four different sets of conditions implying the SSD between the order statistics, that are expressed in terms of the ranks $i$, $j$ and the sample sizes $m$, $n$. Moreover, this approach to finding conditions implying the SSD order may be extended to the two-sample problem by assuming that the two parent distributions are ordered w.r.t. a fractional-degree stochastic dominance relation recently introduced by Lando and Bertoli-Barsotti~\cite{landodist}. This method provides a flexible framework for SSD comparisons according to the information available on the parent distribution's shape. Finally, to draw such information from data, we propose statistical tests to evaluate whether the parent distribution is $H^{-1}$--convex.


\section{Preliminaries}
\noindent We consider absolutely continuous RVs with finite means. Let us begin with some notations. Let $X$ be an RV with CDF $F_X$ and probability density function (PDF) $f_X$. For any ordering relation $\succ$ we shall write $X\succ Y$ or $F_X \succ F_Y$ interchangeably. Let $X_{k:n}$ the $k$-th order statistic corresponding to an i.i.d. random sample of size $n$ from $X\sim F_X$. It is well known that the CDF of $X_{k:n}$ is given by  $F_{B} \circ F_X$, where $B\sim beta(k,n-k+1)$, Jones~\cite{jones2004}. Expressions are strongly simplified if we consider sample minima and maxima, whose CDFs reduce to $1-(1-F_X)^k$ and  $F_X^k$, respectively. But, in general, investigating stochastic orders between order statistics is quite complicated, owing to the number of parameters and non-closed functional forms. We recall the basic definitions of FSD and SSD.
\begin{definition}\label{FSD}
We say that $X$ dominates $Y$ w.r.t. FSD and we write $X \ge_1 Y$ iff $F_{X} (x)\le F_{Y} (x)$, for every $x \in \mathbb{R}$. Equivalently, $X \ge_1 Y$ iff $E(g(X))\geq E(g(Y))$ for every increasing function $g$.
\end{definition}
\begin{definition}
We say that $X$ dominates $Y$ w.r.t. SSD and we write $X \ge_2 Y$ iff $\int_{-\infty}^{x}{F_{X}(t)dt}\leq \int_{-\infty}^{x}{F_{Y}(t)dt}$, for every $x\in \mathbb{R}$. Equivalently, $X \ge_2 Y$ iff $E(g(X))\geq E(g(Y))$ for every increasing concave function $g$.
\end{definition}
Intuitively, FSD represents preference for the RV with larger magnitude, as Definition~\ref{FSD} says that $X$ is less likely than $Y$ to take values in any left tail. On the other hand, SSD represents preference for the RV with larger magnitude or smaller dispersion: in particular, if $X$ dominates $Y$ w.r.t. SSD then $E(X)\geq E(Y)$, and, in case of equality, $\mathrm{Var}(X)\leq \mathrm{Var}(Y)$ and $\mathrm{\Gamma}(X)\leq \mathrm{\Gamma}(Y)$, where $\mathrm{\Gamma}$ is the Gini coefficient, Fishburn~\cite{fishburn1980} or Muliere and Scarsini~\cite{muliere1989}.
As discussed in the literature, in many practical situations it is convenient to use orders that interpolate FSD and SSD, defining a family of fractional-degree dominance relations between these two. We refer to Fishburn~\cite{fishburn1976}, M\"uller adn Scarsini~\cite{muller2017} or Huang et al.~\cite{huang2020} and, in particular, to Lando and Bertoli-Barsotti~\cite{landodist}, who achieved this objective by comparing sample maxima through SSD for a fixed sample size, that is, $X_{k:k}\geq_2 Y_{k:k}$, for some positive integer $k$. For technical reasons, such an order proves useful in deriving SSD conditions for order statistics in the two-sample problem, as we show in Section 3. Moreover, \cite{landodist} proved that $X \ge_{1} Y$ iff $X_{k:k}\geq_2 Y_{k:k}$, for every positive integer $k$, (whereas, clearly, $X \ge_{2} Y$ iff $X_{1:1}\geq_2 Y_{1:1}$). In general, assuming $k\geq h\geq 1$, the following relations hold
\begin{equation}
X \geq_{1} Y\quad\Rightarrow\quad X_{k:k}\geq_2 Y_{k:k} \quad\Rightarrow\quad X_{h:h}\geq_2 Y_{h:h}\quad\Rightarrow\quad X \geq_{2} Y.
\end{equation}
The implications above mean that, if  $X\geq_2 Y$ is satisfied, the relation $X_{h:h}\geq_2 Y_{h:h}$ may hold for every $h$ up to some value $k$. Moreover, as larger values of $h$ correspond to stronger orders, this justifies the definition of a fractional-degree stochastic dominance relation, introduced in \cite{landodist},
\begin{equation}
X\geq_{1+\frac{1}{k}} Y \quad\Leftrightarrow\quad
k=\sup\{h\geq 1: X_{h:h}\geq_2 Y_{h:h}\}.
\end{equation}
For instance, $k=10$ gives $X\geq_{1.1} Y$, meaning that
$X_{h:h}\geq_2 Y_{h:h}$ for $h\leq10$, and $X_{h:h}\not\geq_2 Y_{h:h}$ for $h\geq 11$.

The conditions for SSD are particularly simple if CDFs are single-crossing or, alternatively, if PDFs are double-crossing.
Let us denote the number of sign changes of a function, $u$, defined on an interval, $I$, with
\begin{center}$S^{-}(u)=\sup{\{S^{-}[u(x_{1}),\ldots
,u(x_{\ell})], x_{1}<x_{2}<\ldots<x_{\ell}\}}$,
\end{center}
where $\ell<\infty$, $x_{i}\in I$, $i=1,\ldots,\ell$, and $S^{-}[y_{1},\ldots ,y_{\ell}]$ is the number of sign changes of
the sequence, $y_{1},\ldots ,y_{\ell}$, where the zero terms are
omitted \cite{shaked1982}.
Sufficient conditions for SSD can be derived as follows.
\begin{lemma}\label{LEM1}
If $S^{-}(F_X-F_Y)\le 1$ and the sign sequence
starts with $-$, then $X\ge _{2} Y$ iff $E(X)\ge
E(Y)$ (Hanoch and Levy~\cite{hanoch1969}).

If $S^{-}(f_X-f_Y)\le 2$ and the sign sequence
starts with $-$, then $X\ge _{2} Y$ iff $E(X)\ge
E(Y)$ (Ramos et al.~\cite{ramos2000}).
\end{lemma}

Throughout the paper we shall frequently use the following preservation result.
\begin{lemma}\label{LEM2}
Let $\phi$ be an increasing and concave function. If $X \geq_2 Y$ then $\phi(X) \geq_2 \phi(Y)$.
\end{lemma}
\begin{proof}
For every increasing concave function $g$, the functional composition $g \!\circ\! \phi$ is increasing concave as well. Then, $X \geq_2 Y$ implies $E(g \!\circ\! \phi(X))\geq E(g \!\circ\! \phi (Y))$, for every increasing concave $g$, that is, $\phi(X) \geq_2 \phi(Y)$, by characterization of SSD.
\end{proof}

We shall also need the following orders. The first one is due to Chan et al.~\cite{chan1990}, whereas the latter is ascribable to van Zwet~\cite{zwet1964}.
\begin{definition}
\label{def:twoorders}
Let $H,G$ be a pair of CDFs.
\begin{enumerate}
\item
Let $H$ be absolutely continuous w.r.t. $G$. We say that $H$ is \textit{more convex}
than $G$ and write $H \geq^*_c G$ iff $H \!\circ\! G^{-1}$ is convex.
\item
We say that $H$ dominates $G$ w.r.t. the \textit{convex transform order} and write $H \geq_c G$ iff $H^{-1} \!\circ\! G$ is convex on the support of $G$ or, equivalently, $G$ is $H^{-1}$--convex.
\end{enumerate}
\end{definition}
In fact, the notion of being more convex mentioned above is a translation into a geometrical interpretation of the relative convexity introduced by Hardy et al.~\cite{HLP52}.
The two orders in Definition~\ref{def:twoorders} are different but closely related: if $H$ and $G$ are defined on the unit interval, $H \geq_c F$ iff $H^{-1} \geq^*_c G^{-1}$. In turn, it can be easily seen that $>^*_c$ is equivalent to the likelihood ratio order, Chan et al.~\cite{chan1990}. Moreover, if $F_X\geq_c F_Y$, then Y is said to be a \textit{convex transform} of $X$, since $F_Y^{-1}\!\circ\! F_X(X)$ has the same distribution as $Y$.

Finally, we present two technical lemmas which will be useful in the next section.
\begin{lemma}\label{LEM3}
Let $B_1\sim beta(a_1,b_1)$ and $B_2\sim beta(a_2,b_2)$.
\begin{enumerate}
\item If $a_1\geq a_2$ and $b_1\leq b_2$ then $\frac{f_{B_1}(x)}{f_{B_2}(x)}$ is increasing.
\item If $a_1\geq a_2$ then $S^{-}(f_{B_1}-f_{B_2})\leq 2$, where the sign sequence starts with $-$.
\end{enumerate}
\end{lemma}

\begin{proof}
Let $\ell(x)=\frac{f_{B_1}(x)}{f_{B_2}(x)}=\beta x^{a_1-a_2}(1-x)^{b_1-b_2}$, where $\beta=\frac{B(a_2,b_2)}{B(a_1,b_1)}$. Differentiating we get $\ell'(x)= \beta x^{a_1-a_2-1}(1-x)^{b_1-b_2-1}((a_2-a_1+b_2-b_1)x+a_1-a_2)$. Thus, $S^{-}(\ell'(x))=S^{-}((a_2-a_1+b_2-b_1)x+a_1-a_2).$
\begin{enumerate}
\item
  If $a_1\geq a_2$ and $b_1\leq b_2$, then, $\ell'(x)\geq 0$, meaning that $\ell$ is increasing.
\item
If $a_1\geq a_2$, then $S^{-}(\ell'(x))\leq 1$ and the sign sequence starts with +, meaning that $\ell$ is either increasing or increasing and then decreasing. Consequently, the conclusion follows from $S^{-}(f_{B_1}-f_{B_2})=S^{-}(\ell-1)$.
\end{enumerate}
\end{proof}

\begin{lemma}\label{P2}
Let $H^{-1}$ be a quantile function and $B_{i,n}\sim beta(i,n-i+1)$, where 
$1\leq i \leq n$.
\begin{enumerate}
\item If $H^{-1}(p)=p$ then $E(H^{-1} \!\circ\! B_{i,n})=\frac{i}{n+1}$.
\item If $H^{-1}(p)=\log{\frac{p}{1-p}}$ then $E(H^{-1} \!\circ\! B_{i,n})=\psi(i) - \psi(n-i + 1)$, where $\psi$ is the digamma function.
\item If $H^{-1}(p)=-\log{(1-p)}$ then $E(H^{-1} \!\circ\! B_{i,n})=\sum_{k=n-i+1}^{n}{\frac{1}{k}}$.
\item If $H^{-1}(p)=\frac{p}{1-p}$ then $E(H^{-1} \!\circ\! B_{i,n})=\frac{i}{n-i}$.
\end{enumerate}
\end{lemma}
\begin{proof}
Expression (1) is trivial. Expression (2) can be found in Birnbaum and Dudman~\cite{birnbaum}. Expression (3) is given in Arnold and Nagaraja~\cite{arnold1991exp}. As for (4), $H^{-1} \!\circ\! B_{i,n}$ is a beta distribution of the second type and the expression of the mean is straightforward.
\end{proof}

\section{Main results}
In this section, we enable SSD comparisons of order statistics based on different decompositions of the CDF of $X_{k:n}$, where the decomposition chosen determines the range of application of the corresponding SSD conditions, as we establish in the following theorems.

\begin{theorem}\label{THM1}
\noindent Let $X \sim F$, $B_1\sim beta(i,n-i+1)$ and $B_2\sim beta(j,m-j+1)$.  Let $H$ be a CDF. If $F$ is $H^{-1}$--convex, $i\geq j$ and $E(H^{-1} \!\circ\! B_1)\geq E(H^{-1} \!\circ\! B_2)$, then $X_{i:n}{\ge }_2 X_{j:m}$.
\end{theorem}
\begin{proof}
From Lemma~\ref{LEM3}, it follows that $i\geq j$ implies $S^{-}(f_{B_1}-f_{B_2})\le 2$, where the sign sequence starts with $-$. Consequently $S^{-}(F_{B_1}-F_{B_2})\leq 1$, where the sign sequence starts also with $-$ \cite{ramos2000}. Similarly, since $H$ is increasing we have $S^{-}(F_{B_1}\!\circ\! H-F_{B_2}\!\circ\! H)\le 1$, where $F_{B_i}\!\circ\! H$ is the CDF of $H^{-1} \!\circ\! B_i$, for $i=1,2$. Taking into account that $i\geq j$ and $E(H^{-1} \!\circ\! B_1)\geq E(H^{-1} \!\circ\! B_2)$, it follows from Lemma~\ref{LEM1} that $H^{-1} \!\circ\! B_1\geq_2 H^{-1} \!\circ\! B_2$. Since $H^{-1} \!\circ\! F$ is convex, $F^{-1} \!\circ\! H$ is concave, Lemma~\ref{LEM2} yields
\begin{equation}\label{decomposition}
F^{-1} \!\circ\! H \!\circ\! H^{-1} \!\circ\! B_1 \geq_2 F^{-1} \!\circ\! H \!\circ\! H^{-1} \!\circ\! B_2,
\end{equation}
where the RVs $F^{-1} \!\circ\! B_i$, for $i=1,2$, have CDFs $F_{B_i} \!\circ\! F$, meaning that $X_{i:n}{\ge }_2 X_{j:m}$.
\end{proof}

As different choices of $H$ may lead to different conditions for SSD, it is natural to wonder whether ordered choices of $H$ (in the sense of $\geq_c$) yield ordered (i.e., stronger/weaker) SSD conditions. The next theorem answers to this question.

\begin{theorem}\label{THM2}
Let $X \sim F$, $B_1\sim beta(i,n-i+1)$ and $B_2\sim beta(j,m-j+1)$, where $i\geq j$. Let $H,G$ be a pair of CDFs such that $H \geq_c G$.
\begin{enumerate}
\item $F$ is $G^{-1}$--convex $\quad\Rightarrow\quad$ $F$ is $H^{-1}$--convex.
\item $E(H^{-1} \!\circ\! B_1)\geq E(H^{-1} \!\circ\! B_2)$ $\quad\Rightarrow\quad$ $E(G^{-1} \!\circ\! B_1)\geq E(G^{-1} \!\circ\! B_2)$.
\end{enumerate}
\end{theorem}
\begin{proof}
\begin{enumerate}
\item
The implication is a straightforward consequence of the transitivity of the convex transform order $\geq_c$. In fact, since  $H^{-1} \!\circ\! G$ and $G^{-1} \!\circ\! F$ are convex, also the composition $H^{-1} \!\circ\! G\!\circ\! G^{-1} \!\circ\! F=H^{-1}\!\circ\! F$ is convex.
\item
The implication follows from the single-crossing argument of Lemma~\ref{LEM1}, that is, $i \geq j$ and $E(H^{-1} \!\circ\! B_1)\geq E(H^{-1} \!\circ\! B_2)$ imply $H^{-1} \!\circ\! B_1\geq_2 H^{-1} \!\circ\! B_2$. Taking into account that $G^{-1} \!\circ\! H$ is concave, Lemma~\ref{LEM2} yields $G^{-1} \!\circ\! B_1\geq_2 G^{-1} \!\circ\! B_2$ and, in particular, $E(G^{-1} \!\circ\! B_1)\geq E(G^{-1} \!\circ\! B_2)$.
\end{enumerate}
\end{proof}

In spite of its simplicity, Theorem~\ref{THM2} provides many useful ways to apply Theorem~\ref{THM1}. The general concept can be summarized as follows. Ideally, we wish to be able to compare as many pairs $(X_{i:n}$,$X_{j:m})$ as possible, according to our knowledge of the parent distribution. If only partial information about the $F$ is available, we can check its $H^{-1}$--convexity w.r.t. some suitable $H^{-1}$ and then apply Theorem~\ref{THM2}. Clearly, the more convex $H^{-1}$ is (in the sense of Definition~\ref{def:twoorders}), the more likely we can apply the method, in that we might choose $H$ so that basically every distribution satisfies $H \geq_c F$. On the other hand, choosing an $H^{-1}$ with a higher degree of convexity yields weaker conditions on the parent distribution but, at the same time, stronger conditions on $i,j,n,m$. In other words, the larger the set of comparable families, the smaller the set of comparable order statistics, and vice versa. As a limiting case, we may consider the situation in which $F$ is assumed to be known, which is clearly the most restrictive condition. However, in this case we can directly compute the expectations of the order statistics (at least numerically) and check whether $i \geq j$ and $E(X_{i:n}){\geq} E(X_{j:m})$: this enables us to rank the largest possible set of pairs $(X_{i:n}$,$X_{j:m})$.

We focus on four partially ordered choices of $H$.
\begin{enumerate}
\item Uniform (trivial case): $H(x)=x$, for $x\in[0,1]$, and $H^{-1}(p)=p$, for $p\in[0,1]$.
\item Logistic: $H(x)=\frac{1}{1 + \mathrm{e}^{-x}}$, for $x\in(-\infty,\infty)$, and $H^{-1}(p)=\log{\frac{p}{1-p}}$, for $p\in(0,1)$ (log-odds, or logit function).
\item Exponential: $H(x)=1 - \mathrm{e}^{-x}$, for $x\in[0,\infty)$, and $H^{-1}(p)=-\log{(1-p)}$, for $p\in[0,1)$.
\item Log-logistic: $H(x)=\frac{x}{1+x}$, for $x\in[0,\infty)$, and $H^{-1}(p)={\frac{p}{1-p}}$, for $p\in[0,1)$ (odds function).
\end{enumerate}
Correspondingly, the sets of $H^{-1}$--convex distributions, that is, those satisfying $H\geq_c F$, determine four classes, defined as follows.
\begin{definition}\label{classes} Let $F$ be a CDF. We say that:
\begin{enumerate}
\item $F\in \mathcal{F}_C$ iff $F$ is convex.
\item $F\in \mathcal{F}_{CL}$ iff $F$ is \textit{logit--convex}, i.e., $\log\frac{F}{1-F}$ is convex.
\item $F\in \mathcal{F}_{IFR}$ iff $-\log(1-F)$ is convex.
\item $F\in \mathcal{F}_{CO}$ iff $F$ is \textit{odds--convex}, i.e., $\frac{F}{1-F}$ is convex.
\end{enumerate}
\end{definition}

Let $X\sim F$ represent the failure time. All the four classes defined above are of interest in terms of reliability properties:
\begin{description}
\item[$\mathcal{F}_{IFR}$] is referred to as the IFR class as $-\log(1-F)$, namely, the \textit{hazard function} of $F$, is convex iff $r(x)=\frac{f(x)}{1-F(x)}=\lim_{\Delta x \rightarrow 0} {\frac{P(X\in(x,x+\Delta x]|X>x)}{\Delta x}} $, namely, the \textit{failure rate} of $F$, is increasing. $\mathcal{F}_{IFR}$ is an important class in reliability theory and contains many relevant models (see Shaked and Shanthikumar~\cite{shaked2007} and references therein).
\item[$\mathcal{F}_C$] contains just distributions with bounded support, as $F\in \mathcal{F}_C$ iff the corresponding PDF, $f$, is increasing. Moreover, the probability of failure within a fixed-width interval increases with time, i.e., the function $P(X\in(x,x+\Delta])$ is increasing in $x$ for every positive $\Delta$.
\item[$\mathcal{F}_{CL}$] contains distributions with unbounded support of the form $(-\infty, c)$, where $c\leq +\infty$ (as the limit of logit function at $0$ is $-\infty$). $F\in \mathcal{F}_{CL}$ iff the \textit{log-odds rate}, that is, the ratio between the failure rate and the CDF, $\frac{r}{F}$, is increasing. Equivalently, note that $\frac{r}{F}=r+r^*$, where $r^*=\frac{f}{F}$ is the \textit{reversed failure rate} of $F$. This class has been studied by Zimmer et al.~\cite{zimmer1998}, Wang et al.~\cite{wang2003}, Sankaran and Jayakumar~\cite{sankaran2008} or Navarro et al.~\cite{navarro2008}.
\item[$\mathcal{F}_{CO}$] is characterized by the convexity of the odds for failure, that is, the ratio between failure and survival probability, where it can be seen that
\begin{center}
$\frac{F(x)}{1-F(x)}=\frac{P(X\in(x,x+\Delta x]|X>x)}{P(X\in(x-\Delta x,x]|X\leq x)},$
\end{center}
for arbitrarily small $\Delta x>0$, Kirmani and Gupta~\cite{kirmani2001}. Equivalently, $F\in \mathcal{F}_{CO}$ iff the ratio between the failure rate and the survival function, $\frac{r}{1-F}$, is increasing. As shown in the sequel, $\mathcal{F}_{CO}$ is the widest class considered, in that it contains all the classes above and also some heavy tailed families.
\end{description}
The relations between these four classes can be derived straightforwardly. First, it can be readily seen that $F$ convex implies $-\log{(1-F)}$ convex. In turn, $\frac{r}{F}$ increasing implies $r$ increasing. Finally, $r$ increasing implies $\frac{r}{1-F}$ increasing. This can be summarized as follows
\begin{equation}\label{class}
\mathcal{F}_C\subset\mathcal{F}_{IFR}\subset\mathcal{F}_{CO}\qquad\text{and}\qquad \mathcal{F}_{CL}\subset\mathcal{F}_{IFR}\subset\mathcal{F}_{CO},
\end{equation}
whereas clearly $\mathcal{F}_C\cap \mathcal{F}_{CL}=\emptyset$, because of the different support assumptions.
The classifications of some basic models are given it Table 1.

\begin{table}[h]
\caption{Different popular distributions and corresponding convexity conditions.}
{\scriptsize\begin{tabular}{|l|l|c|c|c|c|c|c|} \hline
\textbf{Distribution} &\multicolumn{1}{c|}{\textbf{CDF}}&\textbf{Parameters}&\textbf{Support}& $\boldsymbol{\mathcal{F}_C}$  & $\boldsymbol{\mathcal{F}_{CL}}$ & $\boldsymbol{\mathcal{F}_{IFR}}$&$\boldsymbol{\mathcal{F}_{CO}}$ \\ \hline
Uniform & $\displaystyle\frac{x-a}{b-a}$&$b>a$ &$[a,b]$& yes&no&yes&yes  \\ \hline
Power function &$\displaystyle{\left(\frac{x}{b}\right)}^a$&$a,b>0$&$[0,b]$& yes &no&yes &yes \\
 & & & & ($a\geq 1$)& & ($a\geq 1$)&($a\geq 1$) \\ \hline
Logistic & $\displaystyle\frac{\mathrm{1}}{\mathrm{1+}{\mathrm{exp} \left(\frac{\mu -x}{\sigma }\right)\ }}$&$\sigma >0$&$\mathbb{R}$ & no&yes&yes&yes \\ \hline
Gumbel& $\displaystyle1-\exp{(-\mathrm{e}^{\frac{x-\mu}{\sigma }} )}$&$\sigma >0$&$\mathbb{R}$ & no&yes&yes&yes \\ \hline

Exponential &  $\displaystyle1-\mathrm{e}^{-ax}$&$a>0$&$[0,\infty)$ & no&no&yes&yes  \\ \hline
Normal & $\displaystyle\frac{1}{2}\mathrm{erf}\left(\frac{x-\mu}{\sqrt{2}\sigma} \right)$ & $\sigma >0$&$\mathbb{R}$ & no&no&yes&yes  \\ \hline
Beta& $\displaystyle\int_0^x{\frac{(1 - t)^{ b-1} t^{a-1}}{B(a,b)}}dt$&$a,b>0$&$[0,1]$ & yes &no&yes &yes \\
 & & & & ($a,b\leq 1$) & & ($a\geq1$) & ($a\geq 1$)\\
\hline
Gamma& $\displaystyle\int_0^x{\frac{\mathrm{e}^{-x/b} x^{a-1}}{b^{a} \Gamma(a)}}dt$&$a,b>0$&$[0,\infty)$ & no&no&yes &yes \\
 & & & & & & ($a\geq1$)& ($a\geq1$)\\ \hline
Weibull & $\displaystyle1-\exp{-(\frac{x}{b})^{a}}$&$a,b>0$&$[0,\infty)$& no&no&yes &yes \\
 & & & & & & ($a\geq1$)& ($a\geq1$)\\\hline
Cauchy & $\displaystyle\frac{1}{2}+\frac{1}{\pi}\arctan{\frac{x-\mu}{\sigma }}$ & $\sigma>0$&$\mathbb{R}$ & no&no&no&yes \\ \hline
Lognormal & $\displaystyle\frac{1}{2}\mathrm{erf}\left(\frac{\log{x}-\mu}{\sqrt{2}\sigma} \right)$ & $\sigma>0$&$[0,\infty) $& no&no&no&yes \\ \hline
Log-logistic & $\displaystyle\frac{1}{1+{(x/b)}^{-a}}$&$a,b>0$&$(0,\infty)$ & no&no&no&yes \\
 & & & & & & & ($a\geq 1$) \\ \hline
Pareto & $\displaystyle1-\left(\frac{b}{x}\right)^{a}$&$a,b>0$&$(b,\infty)$ & no&no&no&yes \\
 & & & & & & & ($a\geq 1$)\\  \hline
\end{tabular}}
\end{table}

According to the classification of Definition~\ref{classes}, we determine four methods to derive SSD, in the one-sample and in the two-sample problems, as stated in the following corollaries. Note that part (4) of Corollary~\ref{COR1} is already proved in Lando and Bertoli-Barsotti~\cite{lando2019}.

\begin{corollary}\label{COR1}
Let $X \sim F$. If $i \geq j$, the following conditions imply $X_{i:n}{\ge }_2 X_{j:m}$
\begin{enumerate}
\item
$F\in\mathcal{F}_C$ and $\frac{i}{n+1} \geq \frac{j}{m+1}$.
\item
$F\in\mathcal{F}_{CL}$ and $\psi(i) - \psi(n - i + 1)\geq \psi(j) - \psi(m - j+ 1)$, where $\psi$ is the digamma function.
\item
$F\in\mathcal{F}_{IFR}$ and $\sum_{k=n-i+1}^{n}{\frac{1}{k}}\geq \sum_{k=m-j+1}^{m}{\frac{1}{k}}$.
\item
$F\in\mathcal{F}_{CO}$ and $i\geq j$ and $\frac{i}{n} \geq \frac{j}{m}$.\\
\end{enumerate}
\end{corollary}
\begin{proof}
As usual, let $B_1\sim beta(i,n-i+1)$ and $B_2\sim beta(j,m-j+1)$. The results can be proved by repeated application of Theorem~\ref{THM1}, each one based on a different choice for the function $H$ in
(\ref{decomposition}), and, for each case, deriving the expressions of $E(H^{-1} \!\circ\! B_1)$ and $E(H^{-1} \!\circ\! B_2)$ through Lemma~\ref{P2}. According to Definition~\ref{classes}, each of the stated four cases correspond to choosing $H$ as the CDF of the uniform, logistic, exponential and log-logistic, respectively.
\end{proof}

Corollary~\ref{COR1} yields SSD between order statistics by imposing conditions both on i) the shape of the parent distribution (different classes) and ii) the ranks and the sample sizes, $i,j,n,m$. As for i), the relations between different classes are depicted in \eqref{class}. On what regards ii), if $i\geq j$ the following implications follow straightforwardly from Theorem~\ref{THM2}:
\begin{eqnarray*}
\frac{i}{n}\geq\frac{j}{m} & \quad\Rightarrow\quad & \sum_{k=n-i+1}^{n}{\frac{1}{k}}\geq \sum_{k=m-j+1}^{m}{\frac{1}{k}};\\
\sum_{k=n-i+1}^{n}{\frac{1}{k}}\geq \sum_{k=m-j+1}^{m}{\frac{1}{k}} & \quad\Rightarrow\quad & \frac{i}{n+1} \geq \frac{j}{m+1};\\
\sum_{k=n-i+1}^{n}{\frac{1}{k}}\geq \sum_{k=m-j+1}^{m}{\frac{1}{k}} & \quad\Rightarrow\quad & \psi(i) - \psi(n - i + 1)\geq \psi(j) - \psi(m - j+ 1).
\end{eqnarray*}

Corollary~\ref{COR1} is useful in a nonparametric context, in which the functional form of $F$ is supposed to be unknown, but some general assumptions on its shape can be made, which can be verified by means of statistical testing, as we discuss in Section~4. Corollary \ref{COR1} can be used \textit{a fortiori} in a parametric context, if the parameters are supposed to be unknown. Finally, if the exact form of $F$ is known (that is a quite unrealistic assumption in statistics), we can directly compare $E(X_{i:n})$ and $E(X_{j:m})$. The next example illustrates some possible applications of Corollary \ref{COR1}.

\begin{example}
Let $F\in \mathcal{F}_{CO}$. Take for instance $n = 200, j = 43, m = 44$. Suppose we need to determine the smallest value of $i$ such that $X_{i:n}{\ge }_2X_{j:m}$. Corollary \ref{COR1} yields $X_{i:200}{\ge }_2X_{43:44}$ for $i\geq\left\lceil {\mathrm{max} \{j,\frac{nj}{m}\}\ }\right\rceil =196$, where $\left\lceil \bullet \right\rceil $ denotes the ceiling function, whereas the values $i=194,i=195$ are not sufficient to guarantee SSD. Now, assume we have additional information about the parent distribution; say $F\in \mathcal{F}_{IFR}$. Since $\sum_{k=l}^{200}{\frac{1}{k}}$ is decreasing in $l$, it is easy to check that $ \sum_{k=7}^{200}{\frac{1}{k}}\geq \sum_{k=2}^{44}{\frac{1}{k}}> \sum_{k=8}^{200}{\frac{1}{k}}$, meaning that $X_{i:200}{\ge }_2X_{43:44}$, for every $i\geq 194$ while we cannot ensure $X_{i:200}{\ge }_2X_{43:44}$ for $i\leq 193$. In a parametric context, if we assume that  $F$ is a Gamma distribution with unknown shape parameter $a\geq 1$ and unknown scale parameter $b$, we know that $X_{i:200}{\ge }_2X_{43:44}$ for $i\geq 194$, as $F\in \mathcal{F}_{IFR}$.
Moreover, if we assume that the parameters are known, e.g. $a=b=2$, we can compute the expectations of the order statistics and surprisingly, we obtain again that SSD holds just for $i\geq 194$, meaning that strong additional assumptions on the parent distribution do not necessarily weaken the conditions on $i$, obtained through Corollary \ref{COR1}.
\end{example}

Let $X_1,\ldots ,X_n$ denote a sample of i.i.d. RVs from an RV $X$ and $Y_1,\ldots ,Y_m$ denote a sample of i.i.d. RVs from another RV $Y$. The following corollary enables the determination of the sample sizes $n\ \mathrm{and}\ m$ and the ranks $i$ and $j$ such that $X_{i:n}{\geq}_2Y_{j:m}$ by introducing an extra condition on $X$ and $Y$. In particular, we need that $X$ dominates $Y$ w.r.t. an order which is easy to verify and is stronger than SSD, namely $X\geq_{1+\frac{1}{k}} Y$, where $k\geq i$ (\textit{a fortiori}, FSD is clearly sufficient).

\begin{corollary}\label{COR2}
Assume that $i\geq j$ and $X\geq_{1+\frac{1}{k}} Y$, where $k\geq i$. Each one of the following conditions imply $X_{i:n}{\geq}_2Y_{j:m}$.
\begin{enumerate}
\item
$F_Y\in\mathcal{F}_C$ and $\frac{i}{n+1} \geq \frac{j}{m+1}$.
\item
$F_Y\in\mathcal{F}_{CL}$ and $\psi(i) - \psi(n - i + 1)\geq \psi(j) - \psi(m - j+ 1)$.
\item
$F_Y\in\mathcal{F}_{IFR}$ and $\sum_{k=n-i+1}^{n}{\frac{1}{k}}\geq \sum_{k=m-j+1}^{m}{\frac{1}{k}}$.
\item
$F_Y\in\mathcal{F}_{CO}$ and $i\geq j$ and $\frac{i}{n} \geq \frac{j}{m}$.\\
\end{enumerate}
\end{corollary}
\begin{proof}
Let $B_1\sim beta(i,n-i+1)$ and $B_2\sim beta(j,m-j+1)$. First, we prove that $i\geq j$ and $n-i\leq m-j$ implies $F_{B_1}\geq^*_c F_{B_2}$. Indeed, $F_{B_1}\geq^*_c F_{B_2}$ iff the likelihood ratio $\ell(x)=\frac{f_{B_1}(x)}{f_{B_2}(x)}$ is increasing, Chan et al.~\cite{chan1990}, which is implied by $i\geq j$ and $n-i\leq m-j$ (see Lemma~\ref{LEM3}).

Condition $X\geq_{1+\frac{1}{k}} Y$ is equivalent to $\int_0^x{(F_X (t))^k}\leq \int_0^x{(F_Y (t))^k}$, for every $x\geq 0$. Therefore, if $i\leq k$, then $P_k>_c^* F_{B_1}$, where $P_k$ is a CDF such that $P_k(t)=t^k$ on the support $[0,1]$. Hence, we apply Theorem 1 of Lando and Bertoli-Barsotti~\cite{landodist}, which establishes that $X_{k:k}\geq_2 Y_{k:k}\Rightarrow X_{i:n}\geq_2 Y_{i:n}$, for every $i\leq k$. Now, if $ i\geq j$ and if any of the conditions 1., 2., 3. or 4. hold, we can apply Corollary~\ref{COR1}, which yields $Y_{i:n}\geq_2 Y_{j:m}$, and the conclusion follows by transitivity.
\end{proof}

Basically, we can compare order statistics with different parent distributions according to the strength of the dominance relation between them. Such strength, determined by $k$, imposes constraints on $i$ ($i \leq k$). FSD enables the comparisons for every value of $i \geq j$ but, on the other hand, it is the strongest order. Put otherwise, all sets of conditions are well balanced: if we relax some constraints, we need to compensate by strengthening some of the others.

Similarly to Corollary~\ref{COR1}, Corollary~\ref{COR2} is suitable for those situations in which the parent distributions are supposed to be unknown: in such cases, the dominance relation between them can be tested nonparametrically. Tests for stochastic dominance of degree $1+1/k$ may be obtained by readapting SSD tests, using the relation $F_X\geq_{1+\frac{1}{k}}F_Y\Leftrightarrow F^k_X\geq_{2}F^k_Y$, however, this is beyond the scope of our paper. In any case, we can always rely on existing tests for FSD.

\begin{example}
Let $X$ be a Dagum RV with CDF $F_X(x)=(1 + \frac{27}{x^3})^{-2}$, for $x>0$ and let $Y$ be a log-logistic RV with parameters $a=b=2$. First we determine the degree of the dominance relation between $X$ and $Y$. By studying the function $F_X-F_Y$, we find that $S^{-}(F_X-F_Y)\le 1$ and the sign sequence starts with $-$ (the two CDFs cross at $x\approx 13.57$). Since the power function is increasing, we have also $S^{-}(F_X^k-F_Y^k)\le 1$ and the sign sequence starts with $-$, so that $X\geq_{1+\frac{1}{k}} Y$ iff $E(X_{k:k})\geq E(Y_{k:k})$, where $E(X_{k:k})=\frac{- \Gamma(-\frac13)\Gamma(\frac13 + 2 k)}{\Gamma(2 k)}$ and $E(Y_{k:k})=\frac{2 \sqrt{\pi} \Gamma(\frac12 + k)}{\Gamma(k)}$, Lin~\cite{lin2000}. Therefore, we obtain $E(X_{k:k})\geq E(Y_{k:k})$ iff $k\leq9$, that is, $X\geq_{1+\frac{1}{9}} Y$. Using the information about $F_Y$, that is $F_Y\in\mathcal{F}_{CO}$ (note that $F_Y$ belongs only to this class), it follows from Corollary~\ref{COR2} that $X_{i:n}{\ge }_2Y_{j:m}$ for $9\geq i\geq j$ and $\frac{i}{n}\geq\frac{j}{m}$. For instance, letting $n=30$, $j=4$, $m=25$, we get $X_{i:30}{\geq}_2Y_{4:25}$ for $5\leq i\leq 9$, although we cannot guarantee SSD for $i>9$. The constraint $i\leq 9$ can be removed only by strengthening the dominance relation between the parent distributions. For instance, if we take $Z$ to be a log-logistic with parameters $a=2$, $b=3$ then it is easy to see that $Z\geq_1 Y$, therefore $X_{i:30}{\geq}_2Y_{5:20}$ holds for every $i\geq 8$.
\end{example}
The following example illustrates that despite having some distribution with unknown parameters, one can derive SSD between their corresponding order statistics. Moreover, we show that additional information about the shape of the parent distribution has an effective impact in determining SSD.
\begin{example}
Let $X\geq_1 Y$, where $Y$ is a logistic RVs with unknown parameters. Let $i = 18$, $n = 200$, $j = 4$, $m = 44$. We can derive SSD from conditions 2., 3. or 4. of Corollary~\ref{COR2} (the support of the logistic is unbounded, so definitely it does not belong to $\mathcal{F}_C$). Nevertheless, $F_Y\in\mathcal{F}_{CO}$ but $\frac{18}{200}<\frac{4}{44}$ (condition 4. does not hold); $F_Y\in\mathcal{F}_{IFR}$ but again $\sum_{k=183}^{200}{\frac{1}{k}} < \sum_{k=41}^{44}{\frac{1}{k}}$ (condition 3. does not hold). Finally, since $F_Y\in\mathcal{F}_{CL}$ and $\psi(18) - \psi(183) >\psi(4) - \psi(41)$, we can ensure that $X_{18:200}{\geq}_2Y_{4:44}$ only through condition 2..
\end{example}

Sufficient SSD conditions described in Corollary~\ref{COR2} can be of multiple use, for instance they can be applied to determine the range of an unknown parameter of a parent distribution, based on dominance constraints, as shown in the example below.

\begin{example}
Suppose we have a $j$-out-of-$m$ system with log-logistic parent distribution $F_Y$, with parameters $a_Y=3$, $b_Y=1$, and we are looking for an $i$-out-of-$n$ system with a log-logistic parent distribution $F_X$ that dominates $F_Y$. Assume that only the scale parameter $b_X$ is known, say $b_X=2$. We need to determine the values of the shape parameter $a_X$ such that $X_{i:n} \geq_2 Y_{j:m}$ for some given instances $i,j,n,m$. Simple algebra shows that $S^- (F_X-F_Y )=1$ with sign sequence $-,+$, for $ a_X\geq a_Y$, hence $X\geq_{1+\frac{1}{i}} Y$ iff $E(X_{i:i})\geq E(Y_{i:i})$, where $E(X_{i:i})=b_X\frac{i\Gamma(\frac{a_X-1}{a_X})\Gamma(\frac1{a_X}+i)}{\Gamma(1+i)}$ (and similarly for $Y$), Lin~\cite{lin2000}. Let, for instance, $i=30$, $j=10$, $n=110$ and $m=100$. Since $F_X,F_Y\in F_{CO}$ and $\frac{i}{n} \geq \frac{j}{m}$. Numerical computation gives $3=a_Y\leq a_X\leq 5.58 \Rightarrow X\geq_{1+\frac{1}{30}} Y$, which, in turn, implies $X_{30:110}\geq_2 Y_{20:100}$.
\end{example}

\section{Testing $H^{-1}$--convexity}
In the literature, various methods have been proposed to test failure rate properties of distributions, with particular reference to the IFR property, Barlow and Proschan~\cite{barlow1969}, Tenga and Santner~\cite{tenga1984}, Bickel~\cite{bickel1969}, Bickel and Doksum~\cite{bickeldoksum}, Proscha nd Pyke~\cite{proschan1967}, or Sahoo and Sengupta~\cite{sahoo2017}. In this section we study a rather general method to test $H^{-1}$--convexity, where, for technical reasons, $H^{-1}$ is a quantile function such that $H^{-1}(0)=H(0)=0$. Denote with $\mathcal{F}_H$ the family of those CDFs that are $H^{-1}$--convex, that is, such that $H^{-1}\!\circ\! F$ is convex. We aim at testing the null hypothesis $\mathcal{H}_0:F \in \mathcal{F}_H$ against the alternative $\mathcal{H}_1:F \notin \mathcal{F}_H$. In particular, as the logit function does not fit the assumption $H^{-1}(0)=0$, we may be interested in checking whether $F\in \mathcal{F}_C$, $F\in \mathcal{F}_{IFR}$ or $F\in \mathcal{F}_{CO}$, in order to derive SSD between order statistics through a nonparameteric approach.

Denote by $F_n$ the empirical CDF of a random sample $\mathbf{X}=(X_1,\ldots,X_n)$ from $F$, that is, $F_n(t)=\frac1n\sum_{i=1}^{n}{\mathbf{1}_{X_i \leq t}}$. Then $H^{-1} \!\circ\! F_n$ converges almost surely to $H^{-1}\!\circ\! F$ on $[0,F^{-1}(1))$. Denote by $\mathbf{x}=(x_1,\ldots,x_n)$ an ordered realization of the random sample $\mathbf{X}$. Following Tenga and Santner~\cite{tenga1984}, our test is based on the distance between $H^{-1} \!\circ\! F_n$ and its greatest convex minorant (GCM) $g$, that is, the largest convex function that does not exceed  $H^{-1} \!\circ\! F_n$, or, formally,
\begin{equation*}
g(x)=\sup{\{\phi(x): \phi \mbox{ is convex and }\phi(y)\leq H^{-1} \!\circ\! F_n(y),\forall y\in [x_1,x_n]\}},
\end{equation*}
To get a more concrete description of the GCM, the sample $\mathbf{x}$ determines a step function $H^{-1} \!\circ\! F_n$. For notational purposes, henceforth we set $h_k=H^{-1} \!\circ\! F_n(X_{k:n})=H^{-1}(\frac{k}{n})$. Let $i<j<n$ be positive integers, and $L_n^{i,j}$ the straight line connecting $(x_{i},h_{i-1})$ and $(x_{j},h_{j-1})$, that is,
\begin{equation*}
L_n^{i,j}(t)=h_{i-1}+\frac{t-x_{i}}{x_{j-1}-x_{i-1}} \left(h_{j-1}-h_{i-1}\right).
\end{equation*}
The GCM $g$ of the step function $H^{-1} \!\circ\! F_n$, corresponding to $\mathbf{x}$, is a piecewise linear function on $[x_1,x_n]$ defined by
\begin{equation*}
g(t)=
\begin{cases} h_1& t=x_1 \\
\min\left\{h_{j-1},\min \{L_n^{i,k}(x_j):1\leq i<j<k\leq n\}\right\} &t=x_j,\,2\leq j<n \\ 
h_{n-1} &t=x_n
\end{cases}
\end{equation*}
and by linear interpolation for $t\in(x_{j-1},x_j)$.
Intuitively, the value of $g$ at $x_j$ is the minimum of the heights of all segments connecting the nodes $(x_i, h_{i-1})$ and $(x_k, h_{k_1})$, where $i<j\leq k$ (see for instance Figure~\ref{f1}).

In a probabilistic setting, the GCM associated to the random sample $\mathbf{X}$ is an estimator of $H^{-1} \!\circ\! F$ under the assumption of convexity.
The test statistic is based on a distance between $H^{-1} \!\circ\! F_n$ and its GCM, $g$. In particular, we consider a weighted Kolmogorov-Smirnov test statistic, that is,
\begin{equation*}
\KS_n(X_1,...,X_n)=\KS_n(\mathbf{X})=\max_{j\in (1,n)}\{{w_j(h_{j-1}-g(X_{j:n}))}\},
\end{equation*}
where the weights $w_j$  are suitably chosen according to $H^{-1}.$
If $H^{-1}$ is the identity we set $w_j=1$. If $H^{-1}$ is convex (that is, if we test odds-convexity or the IFR property), $h_{j-1}\geq g(X_{j:n})$ and the distance $h_j-h_{j-1}$ is increasing, for $1\leq j\leq n$. Therefore, the weights are tailored to downsize the effect of larger differences due to larger $j$'s. In particular, we set $w_j=\frac{1}{h_{j-1}}$, which provided the best performance in our analysis. Note also that $\KS_n$ is scale invariant.

We reject the null hypothesis for large values of $\KS_n$. The critical values or the $p$-values of the test may be obtained via the least favorable distribution of the test statistic under $\mathcal{H}_0$, which can be determined following the same approach of Tenga and Santner~\cite{tenga1984}.
\begin{proposition}\label{P1}
Let $\mathbf{Y}=(Y_{1},\ldots,Y_{n})$ be a random sample from $Y\sim H$, where $H(0)=0$. Under $\mathcal{H}_0:F \in \mathcal{F}_H$, $\KS_n(\mathbf{Y})\geq_1 \KS_n(\mathbf{X})$.
\end{proposition}
\begin{proof}
Let $\mathbf{y}=(y_1,\ldots,y_n)$ be an ordered random sample from $H$ and denote the corresponding empirical CDF by $H_n$. Subsequently, the GCM of $H^{-1} \!\circ\! H_n$ is denoted by $g^*$.
Let $x_i=v(y_i)=F^{-1}\!\circ\! H(y_i)$, for $i=1,\ldots,n$. Then $\mathbf{x}=(x_1,\ldots,x_n)$ is an ordered random sample from $F$. It is sufficient to show $\KS_n(\mathbf{y})\geq\KS_n(\mathbf{x})$ for any ordered vector $\mathbf{y}$. For both vectors, we have $H^{-1}\!\circ\! H_n(y_i)=H^{-1}\!\circ\! F_n(x_i)=H^{-1}(\frac{i}{n})=h_i$, for $i=1,\ldots,n$.
Under $\mathcal{H}_0$, the function $v=F^{-1}\!\circ\! H$ is increasing and concave with $v(0)=0$. Hence, Theorem~2.2 of Tenga and Santner~\cite{tenga1984} yields $g(y_{i}) \geq g^*(x_{i})$, for $i=1,\ldots,n$, which implies the conclusion of the statement.
\end{proof}



The least favorable distribution of $\KS_n$ under $\mathcal{H}_0$ can be computed through simulation. Let $\mathbf{x}$ be a realization of $\mathbf{X}$. We reject $\mathcal{H}_0$ when $\KS_n(\mathbf{x}) \geq c_{\alpha,n}$, where $c_{\alpha,n}$ is the solution of $P(\KS_n(\mathbf{Y})\geq c_{\alpha,n}) \geq \alpha$ and $\alpha$ is the size of the test. Alternatively, we can compute the $p$-value of the test, that is, $p=P(\KS_n(\mathbf{Y})\geq \KS_n(\mathbf{x}))$.

\subsection{Simulations}
In this subsection we focus on the null hypothesis $\mathcal{H}_0:F\in \mathcal{F}_{CO}$. Likewise, we may obtain tests for $\mathcal{H}_0:F\in \mathcal{F}_{C}$ or $\mathcal{H}_0:F\in \mathcal{F}_{IFR}$ (as for the IFR class, the reader is referred to Tenga and Santner~\cite{tenga1984} and to the aforementioned literature).
The computational work has been performed in Mathematica~\cite{mathematica}.
Compatibly with our computing capacities, we generated 3000 random samples of sample sizes $n=10,15,20,25,30$ and 1000 random samples of size $n=40,50,75,100$. Correspondingly, we obtained quantiles of the test statistic $\KS_n$, reported in Table~2, which may be used to determine (approximately) critical values and $p$-values. These tests are conservative, thus, as rule of thumb, values of the test statistic above 0.9 may provide evidence against $\mathcal{H}_0$ (for all the sample sizes considered). Simulation studies confirm the validity of the test.
\begin{enumerate}
\item
We simulated 100 random samples from a gamma distribution with different shape parameters $a=2,1,0.5$ and $b=1$ (the test is scale invariant). The test yields large $p$-values and 100\% acceptance rate for $a=2$ and in the limit case $a=1$. Nevertheless, for $a=0.5$, although $p$-values decrease, the rejection rate is quite low (always less than 50\% for $\alpha=0.1$): this is due to the shape of the function $\frac{F}{1-F}$, which is actually convex for $x\gtrapprox 0.16$ (as it can be seen from the second derivative). Therefore, the performance of the test is satisfactory, even in the latter case, as the median of a gamma(0.5,1) is approximately $0.22$, which implies that more than half of the observations are expected to exceed 0.16 (so that it is likely that most of the nodes of $\frac{F_n}{1-F_n}$ lay on a convex curve).
\item
We simulated 100 random samples from a Pareto distribution and set $a=2,1,0.5$. The test yields large $p$-values and high acceptance rates for $a=2$. For $a=0.5$, $p$-values are low and the rejection rate is high (always more than 65\% for $\alpha=0.1$) and increases with sample size. Therefore, the simulated power of the test has an increasing trend, as $n$ increases. Moreover, in the limiting case $a=1$, the $p$-values average around 0.5 and exhibit largest standard deviation (around 0.3). This is quite logical: in fact, the CDF of the Pareto(1,1), $F$, is just a shifted version of $H(x)=\frac{x}{1+x}$, that is, $F(x)=H(x+1)$, thus $F$ and $H$ are equivalent w.r.t. $\geq_c$.
\end{enumerate}
The results of the simulations are reported in Table~3. Some examples of the construction of GCM are depicted in Figures~1~and~2.

\begin{table}[h]
\caption{Simulated quantiles of $\KS_n$ for the test $\mathcal{H}_0:F\in \mathcal{F}_{CO}$. For each value of $n$, the number of simulation runs is given at the top.}
{\scriptsize
\begin{tabular}{|l|r|r|r|r|r|r|r|r|r|} \hline
$Runs$ & \multicolumn{1}{c|}{$3000$} & \multicolumn{1}{c|}{$3000$} & \multicolumn{1}{c|}{$3000$} & \multicolumn{1}{c|}{$3000$} & \multicolumn{1}{c|}{$3000$} & \multicolumn{1}{c|}{$1000$} & \multicolumn{1}{c|}{$1000$} & \multicolumn{1}{c|}{$1000$} & \multicolumn{1}{c|}{$1000$} \\ \hline
\backslashbox{$p$}{$n$}& \multicolumn{1}{c|}{$n=10$} & \multicolumn{1}{c|}{$n=15$} & \multicolumn{1}{c|}{$n=20$} & \multicolumn{1}{c|}{$n=25$} & \multicolumn{1}{c|}{$n=30$} & \multicolumn{1}{c|}{$n=40$} & \multicolumn{1}{c|}{$n=50$} & \multicolumn{1}{c|}{$n=75$} & \multicolumn{1}{c|}{$n=100$} \\ \hline
$p=0.1$&$0.510$&$0.561$&$0.597$&$0.599$&$0.605$&$ 0.585 $&$0.582$&$0.603$&$0.600$ \\ \hline
$p=0.2$&$0.639$&$0.677$&$0.698$&$0.700$&$0.704$&$ 0.694 $&$0.694$&$ 0.715 $&$0.713$ \\ \hline
$p=0.3$&$0.734$&$0.752$&$0.777$&$0.777$&$0.768$&$ 0.776 $&$0.711$&$ 0.776 $&$0.775$ \\ \hline
$p=0.4$&$0.805$&$0.819$&$0.834$&$0.836$&$0.828$&$  0.833$&$0.838$&$0.840  $&$0.838$ \\ \hline
$p=0.5$&$0.860$&$0.870$&$0.876$&$0.882$&$0.874$&$0.877  $&$0.885$&$ 0.881 $&$0.882$ \\ \hline
$p=0.6$&$0.901$&$0.910$&$0.914$&$0.919$&$0.912$&$ 0.916 $&$0.921$&$  0.926 $&$0.921$ \\ \hline
$p=0.7$&$0.938$&$0.942$&$0.944$&$0.949$&$0.943$&$ 0.946 $&$0.948$&$0.953   $&$0.949$ \\ \hline
$p=0.8$&$0.967$&$0.966$&$0.969$&$0.971$&$0.969$&$ 0.969 $&$0.971$&$0.973   $&$0.972$ \\ \hline
$p=0.9$&$0.987$&$0.987$&$0.988$&$0.988$&$0.987$&$0.987 $&$0.989$&$0.989   $&$0.989$ \\ \hline
$p=0.95$&$0.994$&$0.994$&$0.995$&$0.995$&$0.995$&$0.995$&$0.995$&$0.995   $&$0.995$ \\ \hline
\end{tabular}}
\end{table}

\begin{table}
\caption{Simulation results ($\mathcal{H}_0:F\in \mathcal{F}_{CO}$). The data refers to 100 simulated samples from 6 different distributions. The cells contain: average $p$-values; standard deviations and acceptance rates ($\alpha=0.1$), separated by semicolons.}
{\footnotesize\begin{tabular}{|l|r|r|r|r|} \hline
\backslashbox{$F$}{$n$}& \multicolumn{1}{c|}{$n=25$} & \multicolumn{1}{c|}{$n=50$} & \multicolumn{1}{c|}{$n=75$} & \multicolumn{1}{c|}{$n=100$} \\ \hline
$gamma(2,1)$&$0.89 ;0.15; 100\%$&$0.94 ;0.12; 100\%$&$0.89 ;0.18; 99\%$&$0.91 ;0.15; 100\%$\\ \hline
$gamma(1,1)$&$0.81 ;0.2; 100\%$&$0.82 ;0.23; 100\%$&$0.82 ;0.21; 99\%$&$0.82 ;0.2; 100\%$\\ \hline
$gamma(0.5,1)$&$0.48 ;0.3; 87\%$&$0.37 ;0.25; 87\%$&$0.27 ;0.25; 72\%$&$0.2 ;0.2; 59\%$\\ \hline
$Pareto(2,1)$&$0.73 ;0.22; 95\%$&$0.81 ;0.23; 99\%$&$0.8 ;0.23; 99\%$&$0.71 ;0.26; 98\%$\\ \hline
$Pareto(1,1)$&$0.5 ;0.3; 89\%$&$0.49 ;0.28; 91\%$&$0.48 ;0.31; 85\%$&$0.53 ;0.32; 89\%$\\ \hline
$Pareto(0.5,1)$&$0.13 ;0.2; 35\%$&$0.09 ;0.12; 25\%$&$0.04 ;0.07; 14\%$&$0.04 ;0.1; 11\%$\\ \hline
\end{tabular}}
\end{table}

\begin{figure}\label{f1}
\begin{center}
\includegraphics{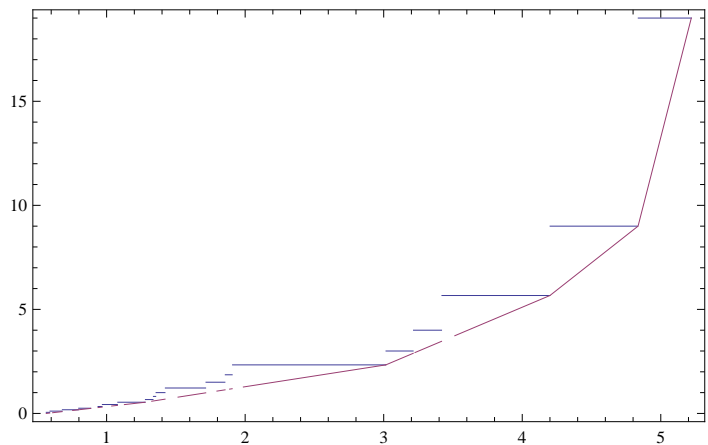}
\caption{$F\sim gamma(2,1)$ and $n=20$. The figure shows $ \frac{ F_n}{1-F_n}$ and its GCM. $\KS_n=0.36$, $p$-value=0.996.}
\end{center}
\end{figure}

\begin{figure}
\begin{center}
\includegraphics{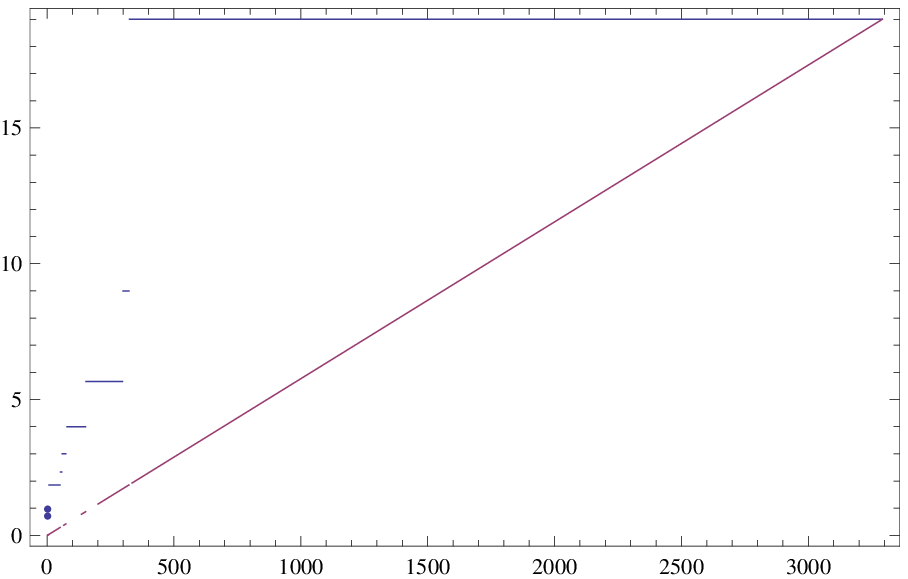}
\caption{$F\sim Pareto(0.5,1)$ and $n=20$. The figure shows $\frac{ F_n}{1-F_n}$ and its GCM. $\KS_n=0.996$, $p$-value=0.06.}
\end{center}
\end{figure}

\section{Discussion}
The method proposed in this paper enables derivation of SSD according to different assumptions on the parent's distribution shape and different conditions to be verified on ranks and sample sizes. We find that stronger (weaker) conditions on $F$ yield weaker (stronger) conditions on $(i,j,n,m)$.

Whereas most of the results in the literature regarding stochastic comparisons of order statistics rely on parametric assumptions, SSD conditions can be easily obtained by testing the hypothesis $F\in \mathcal{F}_{H}$. In the two-sample problem, we can determine whether the order statistics preserve an existing dominance relation between the parent distributions by using the same approach. The method has a wide range of application, as the four classes analyzed make it possible to compare most basic models. As shown in some examples, there are different ways to apply our results.

Finally, it should be stressed that the present study can be easily extended, in order to derive the \textit{increasing and convex order} Shaked and Santhikumar~\cite{shaked2007}, an order that is somewhat complementary to SSD. In doing so, we would deal with different classes of parent distributions, such as the \textit{decreasing failure rate} class.

\bibliographystyle{plain}      
\bibliography{biblio}   

\begin{thebibliography}{10}

\bibitem{arnold1991exp}
Barry~C Arnold and HN~Nagaraja.
\newblock Lorenz ordering of exponential order statistics.
\newblock {\em Statistics \& probability letters}, 11(6):485--490, 1991.

\bibitem{arnold1991}
Barry~C Arnold and Jose~A Villase{\~n}or.
\newblock Lorenz ordering of order statistics.
\newblock {\em Stochastic orders and decision under risk. IMS Lecture
  Notes-Monograph Series}, pages 38--47, 1991.

\bibitem{barlow1963}
Richard~E Barlow, Albert~W Marshall, and Frank Proschan.
\newblock Properties of probability distributions with monotone hazard rate.
\newblock {\em The Annals of Mathematical Statistics}, 34(2):375--389, 1963.

\bibitem{barlow1969}
Richard~E Barlow and Frank Proschan.
\newblock A note on tests for monotone failure rate based on incomplete data.
\newblock {\em The Annals of Mathematical Statistics}, 40(2):595--600, 1969.

\bibitem{bickeldoksum}
Peter~J Bickel and Kjell~A Doksum.
\newblock Tests for monotone failure rate based on normalized spacings.
\newblock {\em The Annals of Mathematical Statistics}, 40(4):1216--1235, 1969.

\bibitem{bickel1969}
PJ~Bickel.
\newblock Tests for monotone failure rate ii.
\newblock {\em The Annals of Mathematical Statistics}, 40(4):1250--1260, 1969.

\bibitem{birnbaum}
Allan Birnbaum and Jack Dudman.
\newblock Logistic order statistics.
\newblock {\em The Annals of Mathematical Statistics}, pages 658--663, 1963.

\bibitem{chan1990}
Wai Chan, Frank Proschan, and Jayaram Sethuraman.
\newblock Convex-ordering among functions, with applications to reliability and
  mathematical statistics.
\newblock {\em Lecture Notes-Monograph Series}, pages 121--134, 1990.

\bibitem{fishburn1976}
Peter~C Fishburn.
\newblock Continua of stochastic dominance relations for bounded probability
  distributions.
\newblock {\em Journal of Mathematical Economics}, 3(3):295--311, 1976.

\bibitem{fishburn1980}
Peter~C Fishburn.
\newblock Stochastic dominance and moments of distributions.
\newblock {\em Mathematics of Operations Research}, 5(1):94--100, 1980.

\bibitem{hanoch1969}
Giora Hanoch and Haim Levy.
\newblock The efficiency analysis of choices involving risk.
\newblock {\em The Review of Economic Studies}, 36(3):335--346, 1969.

\bibitem{HLP52}
G.H. Hardy, J.E. Littlewood, and G.~Pólya.
\newblock {\em Inequalities}.
\newblock Cambridge Univ. Press, 1952.

\bibitem{huang2020}
Rachel~J Huang, Larry~Y Tzeng, and Lin Zhao.
\newblock Fractional degree stochastic dominance.
\newblock {\em Management Science}, 2020.

\bibitem{jones2004}
MC2065642 Jones.
\newblock Families of distributions arising from distributions of order
  statistics.
\newblock {\em Test}, 13(1):1--43, 2004.

\bibitem{kirmani2001}
SNUA Kirmani and Ramesh~C Gupta.
\newblock On the proportional odds model in survival analysis.
\newblock {\em Annals of the Institute of Statistical Mathematics},
  53(2):203--216, 2001.

\bibitem{kochar2006}
Subhash Kochar.
\newblock Lorenz ordering of order statistics.
\newblock {\em Statistics \& probability letters}, 76(17):1855--1860, 2006.

\bibitem{kochar2012}
Subhash Kochar.
\newblock Stochastic comparisons of order statistics and spacings: a review.
\newblock {\em ISRN Probability and Statistics}, 2012, 2012.

\bibitem{kochar2014}
Subhash Kochar and Maochao Xu.
\newblock On the skewness of order statistics with applications.
\newblock {\em Annals of Operations Research}, 212(1):127--138, 2014.

\bibitem{landodist}
Tommaso Lando and Lucio Bertoli-Barsotti.
\newblock Distorted stochastic dominance: a generalized family of stochastic
  orders.
\newblock {\em arXiv preprint arXiv:1909.04767}, 2019.

\bibitem{lando2019}
Tommaso Lando and Lucio Bertoli-Barsotti.
\newblock Second-order stochastic dominance for decomposable multiparametric
  families with applications to order statistics.
\newblock {\em Statistics \& Probability Letters}, page 108691, 2019.

\bibitem{lillo2001}
Rosa~E Lillo, Asok~K Nanda, and Moshe Shaked.
\newblock Preservation of some likelihood ratio stochastic orders by order
  statistics.
\newblock {\em Statistics \& probability letters}, 51(2):111--119, 2001.

\bibitem{lin2000}
Chien-Tai Lin.
\newblock A note on the recurrence relations between moments of order
  statistics from right truncated log-logistic distribution.
\newblock {\em Statistical Papers}, 41(1):99--107, 2000.

\bibitem{muliere1989}
Pietro Muliere and Marco Scarsini.
\newblock A note on stochastic dominance and inequality measures.
\newblock {\em Journal of Economic Theory}, 49(2):314--323, 1989.

\bibitem{muller2017}
Alfred M{\"u}ller, Marco Scarsini, Ilia Tsetlin, and Robert~L Winkler.
\newblock Between first-and second-order stochastic dominance.
\newblock {\em Management Science}, 63(9):2933--2947, 2017.

\bibitem{navarro2008}
Jorge Navarro, Jose~M Ruiz, and Yolanda Del~Aguila.
\newblock Characterizations and ordering properties based on log-odds
  functions.
\newblock {\em Statistics}, 42(4):313--328, 2008.

\bibitem{proschan1967}
Frank Proschan and Ronald Pyke.
\newblock Tests for monotone failure rate.
\newblock In {\em Fifth Berkley Symposium}, volume~3, pages 293--313, 1967.

\bibitem{ramos2000}
Hector~M Ramos, Jorge Ollero, and Miguel~A Sordo.
\newblock A sufficient condition for generalized lorenz order.
\newblock {\em Journal of Economic Theory}, 90(2):286--292, 2000.

\bibitem{sahoo2017}
Shyamsundar Sahoo and Debasis Sengupta.
\newblock Testing the hypothesis of increasing hazard ratio in two samples.
\newblock {\em Computational Statistics \& Data Analysis}, 114:119--129, 2017.

\bibitem{sankaran2008}
PG~Sankaran and K~Jayakumar.
\newblock On proportional odds models.
\newblock {\em Statistical papers}, 49(4):779--789, 2008.

\bibitem{shaked1982}
Moshe Shaked.
\newblock Dispersive ordering of distributions.
\newblock {\em Journal of Applied Probability}, 19(2):310--320, 1982.

\bibitem{shaked2007}
Moshe Shaked and J~George Shanthikumar.
\newblock {\em Stochastic orders}.
\newblock Springer Series in Statistics, 2007.

\bibitem{tenga1984}
Robert Tenga and Thomas~J Santner.
\newblock Testing goodness of fit to the increasing failure rate family.
\newblock {\em Naval research logistics quarterly}, 31(4):617--630, 1984.

\bibitem{zwet1964}
Willem~Rutger van Zwet.
\newblock {\em Convex transformations of random variables}.
\newblock Mathematisch Centrum, 1964.

\bibitem{wang2003}
Yao Wang, Anwar~M Hossain, and William~J Zimmer.
\newblock Monotone log-odds rate distributions in reliability analysis.
\newblock {\em Communications in Statistics: Theory and Methods},
  32(11):2227--2244, 2003.

\bibitem{wilfling1996c}
Bernd Wilfling.
\newblock Lorenz ordering of power-function order statistics.
\newblock {\em Statistics \& probability letters}, 30(4):313--319, 1996.

\bibitem{mathematica}
Inc. Wolfram~Research.
\newblock Mathematica, {V}ersion 12.0, 2019.
\newblock Champaign, IL.

\bibitem{wu2020}
Jintang Wu, Mengshou Wang, and Xiaohu Li.
\newblock Convex transform order of the maximum of independent weibull random
  variables.
\newblock {\em Statistics \& Probability Letters}, 156:108597, 2020.

\bibitem{zimmer1998}
William~J Zimmer, Yao Wang, and Pramod~K Pathak.
\newblock Log-odds rate and monotone log-odds rate distributions.
\newblock {\em Journal of quality technology}, 30(4):376--385, 1998.

\end{thebibliography}

%
%

\end{document}